\documentclass[12pt]{article}

\usepackage[affil-it]{authblk}  
\usepackage[round]{natbib} 
\usepackage{tikz}

\usepackage{graphicx}
\usepackage[colorlinks, linkcolor=blue, citecolor=blue]{hyperref}

\usepackage{amssymb,amsthm,amsmath}
\usepackage{multirow,bigstrut}
\usepackage{subfigure}
\usepackage{enumitem}
\usepackage[ruled,lined]{algorithm2e}
\usepackage{fullpage}

\newcommand{\alg}{\textsc{Alg}\xspace}
\newcommand{\opt}{\textsc{opt}\xspace}
\newcommand{\GR}{\textsc{Greedy}\xspace}
\newcommand{\GF}{\textsc{GreedyFavorite}\xspace}
\newcommand{\GGF}{\textsc{GreedyOrGreedyFavorite}\xspace}
\newcommand{\AU}{\textsc{Assign-U}\xspace}

\newcommand{\CM}{\mathcal{M}}
\newcommand{\CJ}{\mathcal{J}}
\newcommand{\good}{\textnormal{good}}

\newcommand{\block}[1]{
	\begin{center}
		{
			\fbox{
				\begin{minipage}[t]{0.95\textwidth}
					#1
				\end{minipage}
			}
		}
	\end{center}
}
\newcommand{\symmetric}[1]{symmetric #1-favorite machines\xspace}

\newcommand{\favorite}{F}
\newcommand{\p}[2][]{p_{{#2}{#1}}} 
\newcommand{\hatp}[2][]{\hat{p}_{{#2}{#1}}} 
\newcommand{\load}[2][]{l_{#1}{(#2)}} 
\newcommand{\sload}[2][]{\ell_{#1}{(#2)}} 
\newcommand{\tload}[2][]{\tilde{l}_{#1}{(#2)}} 
\newcommand{\tp}[2][]{\tilde{p}_{{#2}{#1}}} 
\newcommand{\oload}[2][]{l^*_{#1}{(#2)}}
\newcommand{\toload}[2][]{\tilde{l}^*_{#1}{(#2)}}
\newcommand{\tsload}[2][]{\tilde{\ell}_{#1}{(#2)}} 

\newtheorem{theorem}{Theorem}
\newtheorem{lemma}[theorem]{Lemma}
\newtheorem{corollary}[theorem]{Corollary}
\newtheorem{observation}[theorem]{Observation}
\newtheorem{claim}[theorem]{Claim}

\newtheorem{example}[theorem]{Example}

\graphicspath{{./figure/}}

\begin{document}

\title{\bf Online scheduling of jobs with favorite machines}

\author[1]{Cong Chen}
\author[2]{Paolo Penna}
\author[1,3]{Yinfeng Xu}
\affil[1]{School of Management, Xi'an Jiaotong University, Xi'an, China}
\affil[2]{Department of Computer Science, ETH Zurich, Zurich, Switzerland}
\affil[3]{The State Key Lab for Manufacturing Systems Engineering, Xi'an, China}

\maketitle

\begin{abstract}
This work introduces a natural variant of the online machine scheduling problem on unrelated machines, which we refer to as the \emph{favorite machine} model. In this model, each job has a minimum processing time on a certain set of machines, called \emph{favorite machines}, and some longer processing times on other machines.
This type of costs (processing times) arise quite naturally in many practical problems.
In the online version, jobs arrive \emph{one by one} and must be allocated irrevocably upon each arrival without knowing the future jobs.
We consider online algorithms for allocating jobs in order to minimize the \emph{makespan}.

We obtain {tight bounds} on the competitive ratio of the greedy algorithm and {characterize} the optimal competitive ratio for the favorite machine model. 
Our bounds generalize the previous results of the greedy algorithm and the optimal algorithm for the \emph{unrelated} machines and the \emph{identical} machines.
We also study a further restriction of the model, called the \emph{symmetric favorite machine} model, where the machines are partitioned equally into two groups and each job has one of the groups as favorite machines. 
We obtain a 2.675-competitive algorithm for this case, and the best possible algorithm for the two machines case.
\end{abstract}

\section{Introduction}
\label{sec:introduction}
Online scheduling on the \emph{unrelated} machines is a classical and well-studied problem.
In this problem, there are $n$ jobs that need to be processed by one of $m$ different machines.
The jobs arrive \emph{one by one} and must be assigned to one machine upon their arrival, without knowing the future jobs.
The time to process a job changes from machine to machine, and the goal is to allocate all jobs so as to minimize the \emph{makespan}, that is, the maximum load over the machines.
The load of a machine is the sum of the processing times of the jobs allocated to that machine. 

Online algorithms are designed to solve the problems when the input is not known from the very beginning but released ``piece-by-piece'' as aforementioned.
Since it is generally impossible to guarantee an optimal allocation,  online algorithms are evaluated through the \emph{competitive ratio}. 
An online algorithm whose competitive ratio is $\rho \geq 1$ guarantees an allocation whose makespan is at most $\rho$ times the optimal makespan.

For scheduling on \emph{unrelated} machines, online algorithms have a rather ``bad'' performance in terms of the competitive ratio.
For example, the simple greedy algorithm has a competitive ratio of $m$, the number of machines;
the best possible online algorithm achieves a competitive ratio of $\Theta(\log m)$.
However, some well-known restrictions (e.g., the \emph{identical} machines and \emph{related} machines) admit a much  better performance, that is, a constant competitive ratio.

In fact, many practical problems are neither as simple as the identical (or related) machines cases, nor so complicated as the general unrelated machines. In a sense, practical problems are somewhat ``intermediate'' as the following examples suggest.
\begin{example}[{Two types of jobs} \citep{two-types}]
\label{ex:two-types}
	This restriction of unrelated machines arises naturally when there are two types of products.
	For instance, consider the production of spare parts for cars.
	The manufacturer may decide to use the machines for the production of spare part 1 to produce spare part 2, and vice versa.
	The machine for spare part 1 (spare part 2, respectively) takes time $p$ for part 1 (part 2, respectively), but it can also 	manage to produce part 2 (part 1, respectively) in time $q>p$.
\end{example}
\begin{example}[{CPU-GPU cluster} \citep{CPU-GPU}]\label{ex:GPU-CPU}
	A graphics processing unit (GPU) has the ability to handle various tasks more efficiently than the central processing unit (CPU). 
	These tasks include video processing, image analysis, and signal processing. 
	Nevertheless, the CPU is still more suitable for a wide range of other tasks. 
	A heterogeneous CPU-GPU system consists of a set $M_1$ of GPU processors and a set $M_2$ of CPU processors. 
	The processing time of job $j$ is $p_{j1}$ on a GPU processor and $p_{j2}$ on a CPU processor. 
	Therefore, some jobs are more suitable for GPU and others for CPU. 
\end{example}

Inspired by these examples, we consider the general unrelated machines case by observing that each job in the system has a certain set of \emph{favorite} machines which correspond to the shortest processing time for this particular job. In Example~\ref{ex:two-types}, the machines for spare parts $1$ are the favorite machines for these parts (processing time $p<q$), and similarly for spare parts $2$. In Example~\ref{ex:GPU-CPU}, we also have two type of machines (GPU and CPU) and some jobs (tasks) have GPUs as favorite machines and others have CPUs as favorite machines. 

\subsection{Our contributions and connections with prior work}
We study the online scheduling problem on what we call the \emph{favorite machine} model.
Denote the processing time of job $j$ on machine $i$ by $\p[i]{j}$ and the minimum processing time of job $j$ by $\p{j} = \min_i \p[i]{j}$.
Thus the set of favorite machines of job $j$ is defined as $\favorite_j = \{i|\  \p[i]{j} = \p{j}\}$ and  the favorite machine model is as follows:
\begin{quote}
(\emph{$f$-favorite machines}) This model is simply the unrelated machine setting when every job $j$ has at least $f$ favorite machines ($|\favorite_j|\geq f$).
The processing time of job $j$ on any \emph{favorite} machine $i \in \favorite_j$ is $\p{j}$,
and on any \emph{non-favorite} machine $i \notin \favorite_j$ is an arbitrary value $\p[i]{j} > \p{j}$.
\end{quote}

This model is motivated by several practical scheduling problems.
Besides the \emph{two types of jobs/machines} problems mentioned in Examples~\ref{ex:two-types} and \ref{ex:GPU-CPU},
the model also captures the features of some real life problems.
For example,
workers with different levels of proficiency for different jobs in manufacturing;
tasks/data transfer cost in cloud computing and so on.

It is worth noting that this model interpolates between the \emph{unrelated} machines where possibly only one machine has minimal processing time for the job ($f=1$) and the \emph{identical} machines case $(f=m)$. 
The $f$-favorite machines setting can also be seen as a ``relaxed'' version of \emph{restricted assignment} problem where each job $j$ can be allocated only to a subset $\favorite_j$ of machines:  the  restricted assignment problem is essentially the case where the processing time of a job on a non-favorite machine is always $\infty$. 

We obtain tight bounds on the \GR{} algorithm\footnote{This algorithm assigns each job to a machine whose load after this job assignment is minimized.} and the well-known \AU{} algorithm (designed for unrelated machines by \citet{aspnes1997line}) for the $f$-favorite machines case,
and show the optimality of the \AU{} algorithm by providing a matching lower bound.
For the \GR{} algorithm, the competitive ratio is $\frac{m+f-1}{f}$,
which generalizes the well-known bounds on the competitive ratio of \GR{} for unrelated machines ($f=1$) and identical machines ($f=m$), that is, $m$ and $2 - \frac{1}{m}$, respectively.
The \AU{} algorithm has the optimal competitive ratio $\Theta (\log \frac{m}{f})$, while it is $\Theta (\log {m})$ for unrelated machines.

\paragraph{Easier instances and the impact of ``speed ratio''.} 
Note that whenever $f=\Theta(m)$, the competitive ratio is constant. In particular, for $f =\frac{m}{2}$, \GR{} has a competitive ratio of $3- \frac{2}{m}$.  We consider the following restriction of the model above such that a finer analysis is possible. 
\begin{quote}
	(\emph{\symmetric{$\frac{m}{2}$}}) All machines are partitioned into \emph{two groups} of equal size $\frac{m}{2}$, and each job has favorite machines as exactly one of the two groups (therefore $f=\frac{m}{2}$ and $m$ is even). 
	Moreover, the processing time on non-favorite machines is $s$ times that on favorite machines, where $s\ge 1$ is the \emph{scaling factor} (the ``speed ratio'' between favorite and non-favorite machines).
\end{quote}

We show that the competitive ratio of \GR is at most $\min\{1+(2- \frac{2}{m}) \frac{s^2}{s+1},~s+(2- \frac{2}{m})\frac{s}{s+1},~3- \frac{2}{m}\}~ (< 3)$.
A modified greedy algorithm, called \GF{},  which assigns each job greedily but only among its favorite machines, has a competitive ratio of $2- \frac{1}{f} + \frac{1}{s}~ (< 3)$.
As one can see, the \GR{} is better than \GF{} for smaller $s$, and \GF{} is better for larger $s$.
Thus we can combine the two algorithms to obtain a better algorithm (\textsc{GGF}) with a competitive ratio of at most $\min \{2+\frac{s^2+s-2}{s+1} , 2+\frac{1}{s}\}~ (\le 2.675)$.
Indeed, we \emph{characterize} the \emph{optimal} competitive ratio for the two machines case. That is, for \emph{symmetric $1$-favorite machines}, we show that the \textsc{GGF} algorithm is $\min  \{ 1+ \frac{s^2}{s+1},1+ \frac{1}{s} \}$-competitive and there is a matching lower bound.

For this problem, our results show the \emph{impact of the  speed ratio $s$} on the competitive ratio. This is interesting because this problem generalizes the case of \emph{two related} machines \citep{epstein2001randomized}, for which $s$ is the speed ratio between the two machines.
The two machines case has also been studied earlier from a game theoretic point of view and compared to the two related machines \citep{fullSPoA,Epstein2010}. 

\newcommand{\onlineproblem}[8]{\filldraw[color=#5!60, fill=#5!5, very thick] #3 rectangle #4 node[black,pos=.5] (#1) {#2}  node[#5,xshift = #7, yshift = #8] (B) {#6} ;}

\newcommand{\oldproblem}[5]{\onlineproblem{#1}{#2}{#3}{#4}{red}{#5}{-2cm}{-1.4cm}}
\newcommand{\newproblem}[5]{\onlineproblem{#1}{#2}{#3}{#4}{blue}{#5}{-2cm}{-1.4cm}}

\newcommand{\extraproblem}[8]{\filldraw[color=#5!60, fill=#5!5, very thick] #3 rectangle #4 node[black,pos=.5] (#1) {#2} node[blue,xshift = #7, yshift = #8] (B) {#6} ;}

\begin{figure}[tbp]	
	\centering	
	\begin{tikzpicture}
		
	\oldproblem{related}{two related machines}{(-2,0)}{(2,1)}{$1+\min\{\frac{s}{s+1},\frac{1}{s}\}$, \quad$1+\min\{\frac{s}{s+1},\frac{1}{s}\}~^{(*)}$}
	
	\newproblem{onesymmetric}{symmetric $1$-favorite}{(-2,2)}{(2,3)}{$1+\min\{\frac{s^2}{s+1},1\}$,\quad$1+\min\{\frac{s^2}{s+1},\frac{1}{s}\}~^{(*)}$}
	
	\newproblem{balancedsymmetric}{symmetric $\frac{m}{2}$-favorite}{(-2,4)}{(2,5)}{$\min\{1+(2- \frac{2}{m}) \frac{s^2}{s+1},~s+(2- \frac{2}{m})\frac{s}{s+1},~3- \frac{2}{m}\}$,\quad$2.675$}
	
	\oldproblem{balanced}{balanced two types}{(-2,6)}{(2,7)}{$3- \frac{2}{m}$,\quad $1 + \sqrt{3} \approx 2.732$}
	\newproblem{halffavorite}{$\frac{m}{2}$-favorite machines}{(-5,8)}{(-1,9)}{$3 - \frac{2}{m}$,\quad $-$}
	
	\oldproblem{twotypes}{two types of machines}{(1,8)}{(5,9)}{$1 + \frac{m-1}{f}$ (for $f \le \frac{m}{2}$), \quad $3.85$}
	
	\newproblem{favorite}{$f$-favorite machines}{(-2,10)}{(2,11)}{$1 + \frac{m-1}{f}$, \quad $\Theta(\log \frac{m}{f})~^{(*)}$}
	
	\oldproblem{unrelated}{unrelated machines}{(-2,12)}{(2,13)}{$m, \quad \Theta(\log m)~^{(*)}$}
		
	\draw [->, line width=.6mm, opacity =.1] 
	(related) edge (onesymmetric) (onesymmetric) edge (balancedsymmetric) (balancedsymmetric) edge (balanced) (balanced) edge (twotypes) (balanced) edge (halffavorite) (halffavorite) edge (favorite) (twotypes) edge (favorite) (favorite) edge (unrelated);
	\end{tikzpicture}
	\caption{Comparison between prior problems and results (in red) and our problems and results (in blue). 
	The two bounds below each problem are the competitive ratios for \GR{} and best-known algorithm, respectively, and the ``${(*)}$'' mark represent the optimality of the best-known algorithm.
	Arrows go from a problem to a more general one, and therefore the upper bounds for the general problem apply to the special one as well.
	 }
	\label{fig:problems}
\end{figure}
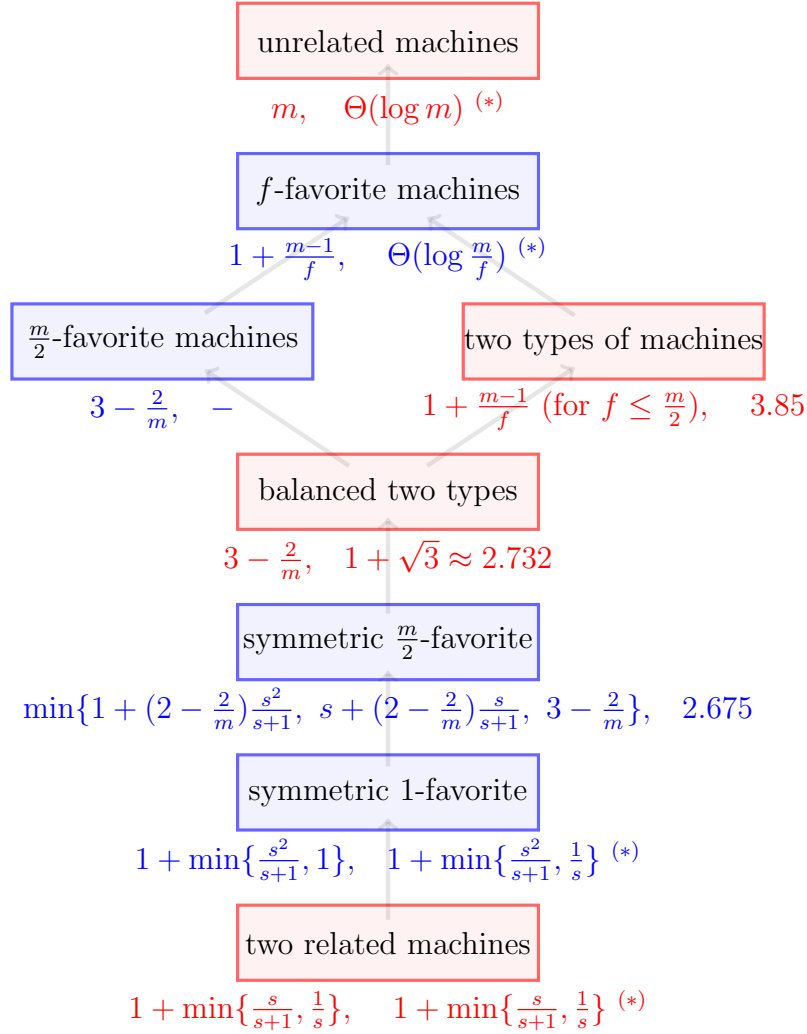

\paragraph{Relations with prior work.}
As mentioned above, the \emph{$f$-favorite machines} is a special case of the \emph{unrelated machines} and a general case of the \emph{identical machines}.
The \emph{symmetric 1-favorite machines} is a generalization of the \emph{2 related machines}.

Besides, there are several other ``intermediate'' problems that have connections with our models and results in the literature (see Figure~\ref{fig:problems} for an overview).
Specifically, in the \emph{two types of machines} case \citep{Imreh2003,CPU-GPU}, they have two sets of machines, $M_1$ and $M_2$, and the machines in each set are identical.
Each job $j$ has processing time $p_{j1}$ for any machine in $M_1$, and $p_{j2}$ for any machine in $M_2$.
The so-called \emph{balanced} case is the restriction in which the number of machines in the two sets are equal. Note that our $\frac{m}{2}$-favorite machines is a generalization of the balanced case, because each job has either $M_1$ or $M_2$ as favorite machines, and thus $f=|M_1|=|M_2|=\frac{m}{2}$.  The \symmetric{$\frac{m}{2}$} case is the restriction of the balance case in which the processing time on non-favorite machines is $s$ times that of favorite ones, i.e., either $p_{j1} = s \cdot p_{j2}$ or $p_{j2} = s \cdot p_{j1}$.

\subsection{Related work} 
\label{sub:related_results}

The \GR algorithm (also known as \textsc{List} algorithm) is a natural and simple algorithm,  often with a provable good performance.
Because of its simplicity, it is widely used in many scheduling problems. In some cases, however, better algorithms exist and \GR is not optimal.

\smallskip
\noindent
\emph{Identical machines.}
For $m$ identical machines, the \GR algorithm  has a competitive ratio exactly $2-1/m$ \citep{graham1966bounds}, and this  is optimal for $m=2,3$ \citep{faigle1989performance}. 
For  arbitrary larger $m$, better online algorithms exist and the bound is still improving \citep{karger1996better,bartal1995new,albers1999better}. Till now the best-known upper bound is 1.9201 \citep{fleischer2000line} and the lower bound is 1.88 \citep{rudin2001improved}.

\smallskip
\noindent
\emph{Related machines (uniform processors).}
The competitive ratio of \GR is ${(1+\sqrt{5})}/{2} \approx 1.618$ for two machines ($m = 2$), and it is at most $1+\sqrt{2m-2}/2$ for $m \ge 3$. The bounds are tight for $m\le 6$  \citep{cho1980bounds}.

When the number of machines $m$ becomes larger, the \GR algorithm is far from optimal, as
its competitive ratio is $\Theta(\log m)$ for arbitrary $m$ \citep{aspnes1997line}.
\citet{aspnes1997line} also devise the \textsc{Assign-R} algorithm with the first constant competitive ratio of 8 for related machines.
The constant is improved to 5.828 by \citet{berman2000line}.

Some research focuses on the dependence of the competitive ratio on speed $s$ when $m$ is rather small.
For the 2 related machines case, \GR has a competitive ratio of $1+\min \{\frac{s}{s+1}, \frac{1}{s}\}$ and there is a matching lower bound \citep{epstein2001randomized}.

\smallskip
\noindent
\emph{Unrelated machines.}
As to the unrelated machines, \citet{aspnes1997line} show that the competitive ratio of \GR is rather large, namely $m$.
In the same work, they present the algorithm \AU{} with a competitive ratio of $\mathcal{O}(\log m)$. A matching lower bound is given by \citet{azar1992competitiveness} in the problem of online restricted assignment.

\smallskip
\noindent
\emph{Restricted assignment.} The online restricted assignment problem is also known as online scheduling subject to arbitrary machine eligibility constraints.
\citet{azar1992competitiveness} show that \GR{} has a competitive ratio less than $\lceil\log_2 m \rceil +1$, and there is a matching lower bound $\lceil \log_2 (m+1) \rceil$.
For other results of online scheduling with machine eligibility we refer the reader to the survey by \citet{Lee2013} and references therein.

\smallskip
\noindent
\emph{Two types of machines (CPU-GPU cluster, hybrid systems).}
\citet{Imreh2003} proves that \GR{} algorithm is $(2+\frac{m_1 -1}{m_2})$-competitive for this case, where $m_1$ and $m_2$ ($\le m_1$) are the number of two sets of machines, respectively. In our terminology, $f= m_2$ and $m_1 = m-f$, meaning that \GR{} is $(1 + \frac{m-1}{f})$-competitive for $f \le \frac{m}{2}$. 
The same work also improves the bound to $4-\frac{2}{m_1}$ with a modified \GR{} algorithm. 
\citet{CPU-GPU} gives a $3.85$-competitive algorithm for the problem, and a simple $3$-competitive algorithm and a more involved $1+\sqrt{3}\approx 2.732$-competitive algorithm for the \emph{balanced case}, that is, the case $m_1=m_2$.

\smallskip
\noindent
\emph{Further models and results.} Some work consider restrictions on the processing times in the offline version of scheduling problems. Specifically, \citet{two-types} consider the case of \emph{two processing times}, where each processing time $p_{ji}\in \{q,p\}$ for all $i$ and $j$.
\citet{kedad2018family} consider the \emph{two types of machines} (CPU-GPU cluster) problem, and
\citet{gehrke2018ptas} consider a generalization, the \emph{few types of machines} problem.
Some work also study similar models in a game-theoretic setting
 \citep{Lavi-Swamy09,Auletta15}. 
Regarding online algorithms, several works consider  \emph{restricted assignment} with additional assumptions on the problem structure like \emph{hierarchical server topologies} \citep{hierarchical}   (see also \citet{CGNP07}).
For other results of processing time restrictions, we refer the reader to the survey by \citet{survey-restricted-assigment}.
Finally, the $f$-favorite machine model in our paper has been recently analyzed in a follow-up paper in the offline game-theoretic setting \citep{CHEN2018}.

\section{Model and preliminary definitions} 
\label{sub:model}
\paragraph{The favorite machines setting.}
There are $m$ machines, $\CM:=\{1,2,\dots,m\}$, to process $n$ jobs, $\CJ:= \{1,2,\dots,n\}$.
Denote the processing time of job $j$ on machine $i$ by $\p[i]{j}$, and the minimum processing time of job $j$ by $\p{j} = \min_{i\in \CM} \p[i]{j}$.
We define the \emph{favorite machines} of a job as the machines with the minimum processing time for the job.
Let $\favorite_j \subseteq \CM$ be the set of favorite machines of job $j$.
Assume that each job has at least $f$ favorite machines, i.e., $|\favorite_j| \ge f$ for $j \in \CJ{}$.
Thus, the processing time of job $j$ on its favorite machines equals to the minimum processing time $\p{j}$ (i.e., $\p[i]{j}= \p{j}$ for $i\in \favorite_j$), while the processing time on its non-favorite machines can be any value that greater than $\p{j}$ (i.e., $\p[i]{j} > \p{j}$ for $i \notin \favorite_j$).

\paragraph{The symmetric favorite machines setting.}
This setting is the following restriction of the favorite machines.  There are $m = 2f$ machines partitioned into two subsets of $f$ machines each, that is, $\CM{}=M_1 \cup M_2$, where $M_1 = \{ 1,2,\dots, f \}$ and $M_2 = \{ f+1,f+2,\dots, 2f \}$.
Each job $j$ has either $M_1$ or $M_2$ as its favorite machines,
i.e., $\favorite_j \in \{ M_1,M_2 \}$. 
The processing time of job $j$ is $\p{j}$ on its favorite machines and $s\cdot \p{j}$ on non-favorite machines, where $s \ge 1$.

\paragraph{Further notation.}
We say that job $j$ is a \emph{good job} if it is allocated to one of its \emph{favorite} machines, and it is a \emph{bad job} otherwise.

Let $\load[i]{j}$ denote the load on machine $i$ after jobs $1$ through $j$ are allocated by online algorithm:
\[
\load[i]{j} = \begin{cases}
\load[i]{j-1} + \p[i]{j} \, , & \text{if job $j$ is assigned to machine $i$}\, ,\\
\load[i]{j-1} \, , & \text{otherwise.}
\end{cases}
\]

In the analysis, we shall sometimes considered the machines in non-increasing order of their loads. For a sequence jobs, we denote by $\sload[i]{j}$ the $i^{th}$ highest load over all machines after the first $j$ jobs are allocated, i.e., 
\begin{align*}
\sload[1]{j} \ge \sload[2]{j} \ge \cdots \ge \sload[m]{j} \, , \quad \text{for any $j \in \CJ$}.
\end{align*}

The jobs arrive \emph{one by one} (over a list) and must be allocated irrevocably upon each arrival without knowing the future jobs. 
the goal is to minimize the \emph{makespan}, the maximum load over all machines.
The \emph{competitive ratio} of an algorithm $A$ is defined as $\rho_A:= \sup_I \frac{C_A(I)}{C_{\opt}(I)},$ where $I$ is taken over all possible sequences of jobs, $C_A(I)$ is the cost (makespan) of algorithm $A$ on sequence $I$, and $C_{\opt}(I)$ is the optimal cost on the same sequence.
We write $C_A$ and $C_{\opt}$ for simplicity whenever the job sequence is clear from the context.

\section{The favorite machine model} 
\label{sec:the_k_favorite_case}
In this section, we first analyze the performance of the  \GR{} algorithm and show its competitive ratio is precisely $\frac{m+f -1}{f}$.
We then show that no online algorithm can be better than $\Omega (\log \frac{m}{f})$ and that algorithm \AU{} \citep{aspnes1997line} has an optimal competitive ratio of $\mathcal O (\log \frac{m}{f})$ for our problem.

\subsection{Greedy Algorithm}  
\label{sub:greedy}
\block{Algorithm \GR{}: Every job is assigned to a machine that minimizes the completion time of this job (the completion time of job $j$ if allocated to machine $i$ is the load $\load[i]{j}$ of machine $i$ after the job is allocated).}

\begin{theorem}
\label{thm.greedyUB}
	The competitive ratio of \GR{} is at most $\frac{m+f-1}{f}$.
\end{theorem}

The key part in the proof of Theorem~\ref{thm.greedyUB} is the following lemma which says that \GR{} maintains the following invariant: the sum of the $f$ largest machines loads never exceeds the sum of all jobs' minimal processing times.

\begin{lemma}\label{lem1}
	For $f$-favorite machines, and for every sequence of $n$ jobs, the allocation of \GR{} satisfies the following condition:
	\begin{equation}\label{eq.induction}
	  \sload[1]{n} + \sload[2]{n} + \cdots + \sload[f]{n} \leq \p{1} + \p{2} + \cdots + \p{n} \,.
	\end{equation}
\end{lemma}

\begin{proof}
	The proof is by induction on the number $n$ of jobs released so far.
	The base case is $n=1$. Since \GR{} allocates the first job to one of its favorite machines, and the other machines are empty, we have  
	$\sum_{i = 1}^f \sload[i]{1} = \sload[i]{1} = \p{1}$, and thus \eqref{eq.induction} holds for $n=1$.

	As for the inductive step, we assume that \eqref{eq.induction} holds for $n-1$, i.e.,
	\begin{equation}\label{eq.basis}
	\sload[1]{n-1} + \sload[2]{n-1} + \cdots + \sload[f]{n-1} \leq \p{1} + \p{2} + \cdots + \p{n-1} \, ,
	\end{equation}
and show that the same condition holds for $n$, i.e., after job $n$ is allocated.

	If job $n$ is allocated as a \emph{good} job,
	then the left-hand side of \eqref{eq.basis} will increase by at most $\p{n}$,
	while the right-hand side will increase by exactly $\p{n}$. Thus, equation \eqref{eq.induction} follows from the inductive hypothesis \eqref{eq.basis}.
	
	If job $n$ is allocated as a \emph{bad} job on some machine $b$, then the following observation allows to prove the statement: before job $n$ is allocated, the load of each favorite machine for job $n$ must be higher than $\load[b]{n-1}$ (otherwise \GR{} would allocate $n$ as a \emph{good} job). Since there are at least $f$ favorite machines for job $n$, there must be a favorite machine $a$ with 
	\begin{equation}\label{eq:gredy:bad-machine-vs-good-machine}
	\load[b]{n-1} < \load[a]{n-1} \leq \sload[f]{n-1}  \ .
	\end{equation}
	Thus, $\load[b]{n-1}$ is not one of the $f$ largest loads before job $n$ is allocated.
	After allocating job $n$, the load of machine $b$ increases to $\load[b]{n} = \load[b]{n-1} + \p[b]{n}$. We then have two cases depending on whether $\load[b]{n}$ is one of the $f$ largest loads after job $n$ is allocated:

	\smallskip
	\noindent
	\textbf{Case 1 ({$\load[b]{n} \le  \sload[f]{n-1}$}).}
	In this case, 
	the $f$ largest loads remain the same after job $n$ is allocated, meaning that the left-hand side of \eqref{eq.basis} does not change, while the right-hand side increases (when adding job $n$).
	In other words,  \eqref{eq.induction} follows from the inductive hypothesis  \eqref{eq.basis}.

	\smallskip
	\noindent
	\textbf{Case 2 ({$\load[b]{n} > \sload[f]{n-1}$}).}
	In this case,  
	$\sload[f]{n-1}$ will be no longer included in the first $f$ largest loads, after job $n$ is allocated, while $\load[b]{n}$ will enter the set of $f$ largest loads: 
	\begin{equation}\label{eq.lk}
	\sum_{~\: i = 1~\:}^f \sload[i]{n} = \sum_{~\: i = 1 ~\:}^f \sload[i]{n-1} - \sload[f]{n-1} + \load[b]{n}\, .
	\end{equation}
	Since \GR{} allocates job $n$ to machine $b$,
	it must be the case
	\begin{equation}\label{eq.sn}
	\load[b]{n} = \load[b]{n-1} + \p[b]{n} \le \load[a]{n-1} + \p{n} \leq \sload[f]{n-1} + \p{n} \, ,
	\end{equation}
	where $a$ is the favorite machine satisfying \eqref{eq:gredy:bad-machine-vs-good-machine}. 
	Substituting \eqref{eq.sn} into \eqref{eq.lk} and by inductive hypothesis \eqref{eq.basis},
	we obtain
	$\sum_{i = 1}^f \sload[i]{n} \leq \sum_{i = 1}^f \sload[i]{n-1}  + \p{n} \leq \sum_{i=1}^{n-1} \p{i} + \p{n}$,
	and thus \eqref{eq.induction} holds. 
	This completes the proof of the lemma.
\end{proof}

\begin{proof}[Proof of Theorem~\ref{thm.greedyUB}]
	Without loss of generality we, assume that the makespan of the allocation of \GR{} is determined by the  last job $n$
	(otherwise, we can consider the last job $n'$ which determines the makespan, and ignore all jobs after $n'$ since they do not increase $C_\GR$ nor decrease $C_\opt$).

	After allocating the last job $n$, the cost of \GR{} becomes
	\begin{align}\label{eq.Cgreedy}
		C_\GR{}  & \le \min \left\{ \min_{i \in \favorite_n} \load[i]{n-1}+ \p{n},\, \min_{i \in \CM{} \backslash 
		\favorite_n} \load[i]{n-1} + \p[i]{n} \right\} \nonumber\\ 
				&  \le \min_{i \in \favorite_n} \load[i]{n-1}+ \p{n} \nonumber\\
				&  \le \sload[f]{n-1} + \p{n} \, ,
	\end{align}
	where the last inequality holds because job $n$ has at least $f$ favorite machines, and thus the least loaded among them must have load at most $\sload[f]{n-1}$.

	Since the optimum must allocate all jobs on $m$ machines, 
	and job $n$ itself requires $\p{n}$ on a single machine, we have
	\begin{align}\label{eq.Copt}
		C_\opt{} \ge \max \left\{ \frac{\sum_{j=1}^n \p{j}}{m} ,\, \p{n} \right\} \ge \max \left\{ \frac{f \cdot \sload[f]{n-1}+ \p{n}}{m} ,\, \p{n} \right\} \, .
	\end{align}
	where the second inequality is due to Lemma~\ref{lem1} applied to the first $n-1$ jobs only  (specifically, we have $\sum_{j = 1}^{n-1} \p{j} \ge \sum_{i = 1}^f \sload[i]{n-1} \geq f \cdot \sload[f]{n-1}$). 
	By combining \eqref{eq.Cgreedy} and \eqref{eq.Copt}, we obtain
	\begin{align*}
		\frac{C_\GR{}}{C_\opt{}} & \le \frac{\sload[f]{n-1} + \p{n}}{\max \left\{ \frac{f \cdot \sload[f]{n-1} + \p{n}}{m} ,\, \p{n} \right\}} \\
		& =\min \left \{\frac{m\left(\frac{\sload[f]{n-1}}{\p{n}}+1\right)}{f\frac{\sload[f]{n-1}}{\p{n}}+1},\,\frac{\sload[f]{n-1}}{\p{n}}+1 \right\}\\
		& \le \frac{m + f - 1}{f} \,, 
	\end{align*}
	where the last inequality holds because the first term decreases in ${\sload[f]{n-1}}/{\p{n}}$ and the second term increases in ${\sload[f]{n-1}}/{\p{n}}$.
\end{proof}

\begin{theorem}
\label{thm.greedyLB}
	The competitive ratio of \GR{} is at least $\frac{m+f-1}{f}$.
\end{theorem}
\begin{proof}
	We will provide a sequence of jobs whose schedule by \GR{} has a makespan $\frac{m+f-1}{f}$, and the optimal makespan is 1.
	For simplicity, we assume that in case of a tie,
	the \GR{} algorithm will allocate the job as a \emph{bad job} to a machine with the \emph{smallest index}.
	Without loss of generality, suppose $m$ is divisible by $f$.  We partition the $m$ machines into $m' := \frac{m}{f}$ groups, $ M_1,M_2,\dots,M_{m'} $, each of them  containing $f$ machines, i.e., $M_i = \{ (i-1) f + 1 ,\, (i-1) f + 2 , \dots ,\, (i-1) f + f \}$ for $i = 1, \dots , m'$.

	The jobs are released in two phases (described in detail below):
	In the first phase, we force \GR{} to allocate each machine in $M_i$ a load of \emph{bad} jobs equal to $i-1$, and a load of good jobs equal to $1 - \frac{i}{s}$ for a suitable $s>1$ (except for $M_{m'}$);
	In the second phase, several jobs with favorite machines in $M_{m'}$ are released and
	contribute an additional $2 - \frac{1}{f}$ to the load of one machine in $M_{m'}$, and thus the makespan is $\frac{m+f-1}{f}$.

	We use the notation $r \times (p, F)$ to represent a sequence of $r$ identical jobs whose favorite machines are $F$, the processing time on favorite machines is $p$, and on non-favorite machines is $s \cdot p$ with $s \ge 1$.
	\begin{figure}[tb]
		\centering
		\includegraphics[width=1\textwidth]{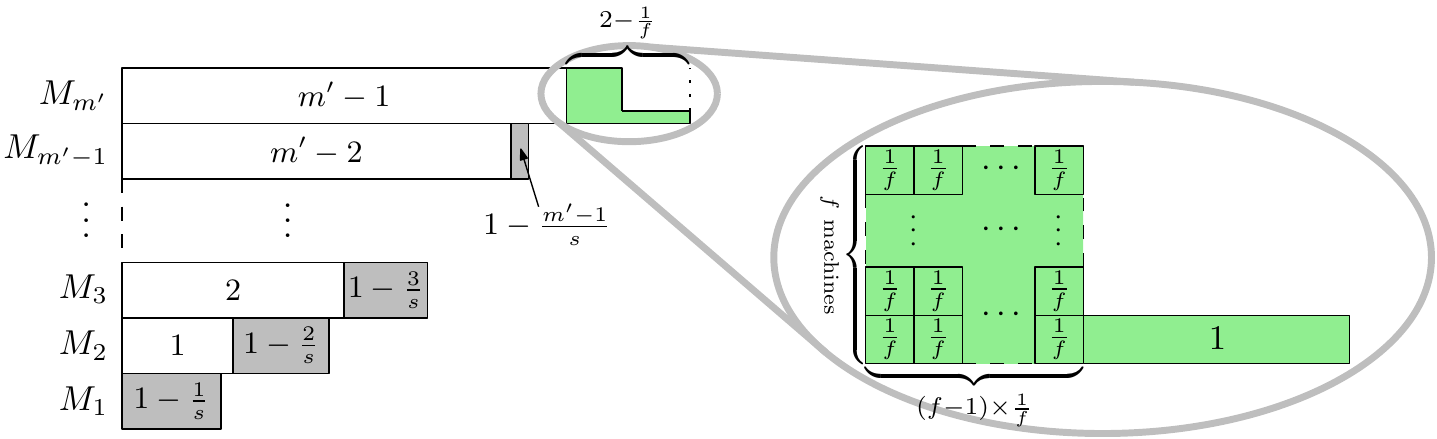}
		\caption{Proof of Theorem~\ref{thm.greedyLB} (bad jobs in white).}
		\label{fig:greedyLB}
	\end{figure}
	
	\smallskip
	\noindent
	\textbf{Phase~1.}
	For each $i$ from $1$ to $m'-1$, release a sequence of jobs $f \times (1- \frac{i}{s},M_i)$ followed by a sequence $f \times (\frac{i}{s},M_i)$.
	In this phase, we take $s > m'-1 + \sqrt{(m'-1)(m'-2)}$, so that,
	for each $i$, the jobs $f \times (1- \frac{i}{s},M_i)$ will be assigned as good jobs to the $f$ machines in $M_i$, one per machine, 
	and the jobs $f \times (\frac{i}{s},M_i)$ will be assigned as bad jobs to the $f$ machines in $M_{i+1}$, one per machine (as shown in Figure~\ref{fig:greedyLB}).
	
	At the end of Phase~1, each machine in $M_i$ ($i \in [1,m'-1]$) will have load  $\frac{i - 1}{s} \cdot s + 1 - \frac{i}{s} = i (1- \frac{1}{s})$,
	and each machine in $M_{m'}$ will have load  $m'-1$.
	The optimal schedule assigns every job to some favorite machine evenly
	so that all machines have load $1$, except for the machines in the last group $M_{m'}$ which are left empty.
	
	\smallskip
	\noindent
	\textbf{Phase~2.}
	In this phase, all jobs released have $M_{m'}$ as their favorite machines,
	specifically $f(f-1) \times (\frac{1}{f}, M_{m'})$ followed by a single job $(1,  M_{m'})$.
	By taking $s > m$,
	no jobs will be allocated as bad jobs by \GR{}.
	Thus, this phase can be seen as scheduling on $f$ identical machines.
	Consequently, all jobs will be allocated as in Figure~\ref{fig:greedyLB} increasing the maximum load of machines in $M_{m'}$ by $2 - \frac{1}{f}$,
	while the optimal schedule can have a makespan 1.
	
	At the end of Phase~2, the maximum load of machines in $M_{m'}$ is $m' - 1 + 2 - \frac{1}{f} = \frac{m+f-1}{f}$, while the optimal makespan is 1.
\end{proof}

\subsection{A general lower bound}  
\label{sub:lower bound}
In this section, we provide a general lower bound for the $f$-favorite machines problem.
The bound borrows the basic idea of the lower bound for the restricted assignment \citep{azar1992competitiveness}.
The trick of our proof is that we always partition the selected machines into arbitrary groups of $f$ machines each, and apply the idea to each of the group respectively.

\begin{theorem}
\label{th:k-favorite-LB}
	Every deterministic online algorithm must have a competitive ratio at least $\frac{1}{2}\lfloor \log_2 \frac{m}{f} \rfloor +1$.
\end{theorem}
\begin{proof}

	Suppose first that $\frac{m}{f}$ is a power of 2, i.e., $\frac{m}{f}=2^{{u} - 1}$. 
	We provide a sequence of $m$ jobs having optimal makespan equal to $1$, while any online algorithm will have a makespan at least $\frac{u+1}{2} = \frac{1}{2}\log_2\frac{m}{f}+1$.
	
	We denote by $r \times (1, F)$ a sequence of $r$ identical jobs whose favorite machines are $F$, the processing time on favorite machines is $1$, and the processing time on non-favorite machines is greater than {$\frac{u+1}{2} = \frac{1}{2}\log_2\frac{m}{f}+1$}, so that all jobs must be allocated to favorite machines.

	We consider subsets of machines $\CM{}_1,\CM{}_2,\ldots, \CM{}_u$, where the first subset is $\CM{}_1= \CM$ and the others satisfy $\CM{}_{i} \subset \CM{}_{i-1}$ and $|\CM{}_i| = \frac{1}{2}|\CM{}_{i-1}| = \frac{m}{2^{i-1}}$ for $i = 2, \ldots u$. 
	We then release $u$ sets of jobs iteratively.
	At each iteration $i$ ($i=1,2,\ldots,u$), a set of jobs $\CJ{}_i$ with favorite machines $\CM{}_u$ are released for allocation, and after the allocation, half of the ``highly loaded'' machines in $\CM{}_i$ are chosen for the next set $\CM{}_{i+1}$.
	More in detail, for each iteration $i$ ($i=1,2,\ldots,u$) we proceed as follows:
	\begin{itemize}[leftmargin=*]
		\item Partition the current set $\CM{}_i$ of machines into arbitrary groups of $f$ machines each, that is, into groups $M_{i,1},M_{i,2}, \dots , M_{i,m_i'}$ where $m_i'=\frac{|\CM{}_i|}{f} = \frac{m}{f \cdot 2^{i-1}} = 2^{{u}-i}$.
		Then, release a set of jobs $\CJ{}_i$:
		\[
		\CJ{}_i :=
		\begin{cases}
			J_{i,1} \cup J_{i,2} \cup \ldots \cup J_{i,m_i'} \,, &\text{ if }1\le i \le {u}-1\,;\\
			\{ f \times (1,\CM{}_{u}) \} \,, &\text{ if }i={u} \,;
		\end{cases}
		\]
		where $J_{i,l}:=\left\{\frac{f}{2}\times (1,M_{i,l}) \right\}$.
		\item The next set  $\CM{}_{i+1}$ of machines consists of the $f/2$ most loaded machines in each subgroup after jobs $\CJ{}_i$ have been allocated. Specifically, we let 
		$			\CM_{i+1}:=M'_{i,1} \cup M'_{i,2} \cup \cdots \cup M'_{i,m_i'}$, for  $M'_{i,l}\subset M{}_{i,l}$ being the subset of $f/2$ most loaded machines in group $M_{i,l}$ after jobs $\CJ{}_i$ have been allocated.
		
	\end{itemize}

	We then show that the following invariant holds at every iteration.
	Note that the following lemma shows a separation of our method and the one of \citet{azar1992competitiveness}.
	Specifically, the maximum load increases by $1/2$ at each iteration in our lower bound, but increases by $1$ in \citet{azar1992competitiveness}.
	\begin{lemma}\label{claim.LB}
		The average load of machines in $\CM{}_i$ is at least $\frac{i-1}{2}$ before the jobs in $\CJ{}_i$ are assigned for $1 \le i \le {u}$.
	\end{lemma}
	\begin{proof}
		The proof is done by induction on $i$.
		For the base case $i=1$, the lemma is trivial. 
		As for the inductive step, suppose the lemma holds for $i$. 

		Denote the average load of $M_{i,l}$ \emph{before} the allocation of jobs $\CJ{}_i$ by $Avg(i,l)$,
		and the average load of $M'_{i,j}$ \emph{after} the allocation of jobs $\CJ{}_i$ by $Avg'(i,l)$.
		We claim that $Avg'(i,l) \ge Avg(i,l)+1/2$ {after} the allocation of jobs $\CJ{}_i$.
		Note that $\CM{}_{i}$ and $\CM{}_{i+1}$ are the union of all groups $M_{i,l}$ and $M_{i,l}'$, respectively.
		Thus, we have that the average load of $\CM{}_{i+1}$ is at least $\frac{i-1}{2} + \frac{1}{2} = \frac{i}{2}$,
		since the average load of $\CM{}_{i}$ is $\frac{i-1}{2}$ (the inductive hypothesis).
		This concludes the proof of the inductive step, and thus the lemma follows.
		Next, we prove that the claim $Avg'(i,l) \ge Avg(i,l)+1/2$ holds:

		\smallskip
		\noindent
		\textbf{Case 1.} There is a machine in $M_{i,l} \backslash M'_{i,l}$ that has load at least $Avg(i,l)+1/2$.
		Since $M'_{i,l}$ consists of the $f/2$ most loaded machines in $M_{i,l}$, each machine in $M'_{i,l}$ will also have load at least $Avg(i,l)+1/2$.
		Thus, we obtain that $Avg'(i,l) \ge Avg(i,l)+1/2$.

		\smallskip
		\noindent
		\textbf{Case 2.} Each machine in $M_{i,l} \backslash M'_{i,l}$ has load no greater than $Avg(i,l)+1/2$.
		Observe that the total load of $M_{i,l}$ is $k \cdot Ave(i,l) + k/2$ after the allocation of $\CJ{}_i$.
		Since $M_{i,l} \backslash M'_{i,l}$ has load at most $k/2 \cdot (Avg(i,l)+1/2)$, the average load of the $f/2$ machines in $M'_{i,l}$ must be at least
		\begin{equation*}
		Avg'(i,l) \ge \frac{ k \cdot Avg(i,l) + \frac{k}{2} - \frac{k}{2}(Avg(i,l)+	\frac{1}{2})}{k/2} = Avg(i,l)+\frac{1}{2} \,. 
		\qedhere
		\end{equation*}			
	\end{proof}
	
	By applying Lemma~\ref{claim.LB} with $i=u$, we have that the average load of $\CM{}_{u}$ is at least $\frac{{u} - 1}{2}$ before the allocation of jobs $\CJ{}_{u}$.
	Thus, after jobs in $\CJ{}_{u}$ are allocated, the average load of $\CM{}_{u}$ is at least $\frac{{u} -1}{2} +1 = \frac{{u} + 1}{2} $, i.e.,~the online cost is at least $\frac{{u} + 1}{2}$.
	
	An optimal cost of $1$ can be achieved by using the machines in $M_{i,l}  \backslash M'_{i,l}$ for $\CJ{}_i$ (for $1 \le i \le {u} - 1$ and $1 \le l \le m_i'$) and using machines in $\CM{}_{u}$ for $\CJ{}_{u}$. Therefore, the competitive ratio is at least $\frac{{u} + 1}{2} =  \frac{1}{2}\log_2 \frac{m}{f} +1 $. Finally, if $\frac{m}{f}$ is not a power of $2$, we simply apply the previous construction to a subset of $m^*$ machines, 
	for $\frac{m^*}{f}= 2^{\lfloor u \rfloor-1 }$ and $u = 1 + \log_2 \frac{m}{f}$. This gives the desired lower bound $\frac{\lfloor u \rfloor + 1}{2} = \frac{1}{2}\lfloor \log_2 \frac{m}{f} \rfloor +1$. 
\end{proof}

\subsection{General upper bound (Algorithm Assign-U is optimal)}  
\label{sub:a_best_possiable_alg}
In this section, we prove a matching  upper bound  for $f$-favorite machines. 
Specifically, we show that the optimal online algorithm for unrelated machines (algorithm \AU{} by \cite{aspnes1997line} described below) yields an optimal upper bound:

\begin{theorem}
\label{thm:AU}
	Algorithm \AU{} can be used to achieve $\mathcal O (\log \frac{m}{f})$ competitive ratio.
\end{theorem}

Following \cite{aspnes1997line}, we show that the  algorithm has an optimal competitive ratio $\rho$ (Lemma~\ref{lem.AU}) for the case the optimal cost $C_\opt{}(\CJ{})$ is \emph{known}, and then apply a standard \emph{doubling} approach to get an algorithm with competitive ratio $4\rho$ for the case the optimum is not known in advance (Theorem~\ref{thm:AU}).
In the following, we use a ``tilde'' notation to denote a normalization by the optimal cost $C_\opt{}(\CJ{})$, i.e., $\tilde{x}=x/C_\opt{}(\CJ{})$.
\block{Algorithm \AU{} \citep{aspnes1997line}: Each job $j$ is assigned to machine $i^*$,
	where $i^*$ is the index minimizing $\Delta_i = a^{\tload[i]{j-1}+ \tp[i]{j}} - a^{\tload[i]{j-1}}$, where $a>1$ is a suitable constant.}
\begin{lemma}\label{lem.AU}
	If the optimal cost is known, \AU{} has a competitive ratio of $\mathcal O (\log \frac{m}{f})$.
\end{lemma}
\begin{proof}
	Without loss of generality, we assume that the makespan of the allocation of \AU{} is determined by the  last job $n$ (this is the same as the proof of Theorem~\ref{thm.greedyUB} above).
	Let $\oload[i]{j}$ be the load of machine $i$ in the optimal schedule after the first $j$ jobs are assigned.
	Define the potential function:
	\begin{equation}\label{eq.potential_function}
		\Phi(j) = \sum_{i=1}^m a^{\tload[i]{j}} \left(\gamma - \toload[i]{j}\right) \, ,
	\end{equation}
	where $a,\gamma > 1$ are constants.
	Similarly to the proof in \cite{aspnes1997line}, we have that the potential function \eqref{eq.potential_function} is non-increasing for  $a = 1+ 1/ \gamma$ (see Appendix~\ref{sub:non-increasing_potential_function} for details).
	Since $\Phi(0) \ge \Phi(n-1)$ and $\toload[i]{n-1}\le 1$, we have
	\begin{equation*}\label{eq.AUpot}
	\gamma \cdot m  \ge  \sum_{i=1}^m a^{\tload[i]{n-1}} \left(\gamma - \toload[i]{n-1}\right) 
	\ge \sum_{i=1}^m a^{\tload[i]{n-1}} (\gamma - 1) 
	\ge f\cdot a^{\tsload[f]{n-1}}(\gamma - 1) \,.
	\end{equation*}
	Thus, it follows that
	\begin{equation}\label{eq.logmk}
	\sload[f]{n-1} \le \log_a \left(\frac{\gamma}{\gamma - 1} \cdot \frac{m}{f} \right) \cdot C_\opt{} \,.
	\end{equation}
	\begin{claim}\label{claim.AU}
		$C_\AU{} \le \sload[f]{n-1} + \p{n}$\,.
	\end{claim}
	\begin{proof}
		Suppose \AU{} assigns the last job $n$ to machine $i'$, it holds that
		\[
		C_\AU{} = \load[i']{n-1} + \p[i']{n} \,,
		\] 
		since we suppose job $n$ is the last completed job.

		\smallskip
		\noindent
		\textbf{Case 1.} Job $n$ is allocated as a good job, i.e., $\p[i']{n} = \p{n}$.
		According to the rule of \AU{}, machine $i'$ must be one of the machines in $\favorite_n$ that has the minimum load.
		Since there are at least $f$ favorite machines, it holds that $\load[i']{n-1} \le \sload[f]{n-1}$. 
		Thus, $C_\AU{} \le \sload[f]{n-1} + \p{n}$.

		\smallskip
		\noindent
		\textbf{Case 2.} Job $n$ is allocated as a bad job, i.e., $\p[i']{n} > \p{n}$.
		Let $i''$ be the machine who has the minimum load in $\favorite_n$. Note that $\load[i'']{n-1} \le \sload[f]{n-1}$.
			According to the rule of \AU{}, it holds that $\load[i']{n-1} \le \load[i'']{n-1}$; otherwise,
			we have $\Delta_{i'} \ge \Delta_{i''}$, which is contradictory to that \AU{} assigns job $n$ to machine $i'$.
			
			Since $\Delta_{i'} \le \Delta_{i''}$ and $\load[i']{n-1} \le \load[i'']{n-1} \le \sload[f]{n-1}$, we have
			\begin{equation*}
			a^{\frac{\load[i']{n-1}+\p[i']{n}}{C_\opt{}}}-a^{\frac{\load[i']{n-1}}{C_\opt{}}} 
			\le a^{\frac{\load[i'']{n-1}+ \p[i'']{n}}{C_\opt{}}}-a^{\frac{\load[i'']{n-1}}{C_\opt{}}} \,,
			\end{equation*}
			so that ${\frac{\load[i']{n-1}+\p[i']{n}}{C_\opt{}}}
			\le {\frac{\load[i'']{n-1}+ \p[i'']{n}}{C_\opt{}}} 
			\le {\frac{\sload[f]{n-1}+ \p[i'']{n}}{C_\opt{}}} \,,$
			implying
			\[
			C_\AU{} = \load[i']{n-1}+\p[i']{n} \le \sload[f]{n-1}+ \p[i'']{n} = \sload[f]{n-1}+ \p{n} \,. \qedhere
			\]
	\end{proof}
	According to Claim~\ref{claim.AU} and \eqref{eq.logmk}, we have
	\begin{align*}
	\frac{C_\AU{}}{C_\opt{}} & \le \frac{\sload[f]{n-1} + \p{n}}{C_\opt{}} 
	\le \log_a \Big(\frac{\gamma}{\gamma - 1} \cdot \frac{m}{f} \Big)  + \frac{\p{n}}{C_\opt{}} \nonumber  \le \log_a \Big(\frac{\gamma}{\gamma - 1} \cdot \frac{m}{f} \Big) + 1 \ .  
	\end{align*}
	Since the latter quantity is $\mathcal O(\log \frac{m}{f})$, this proves Lemma~\ref{lem.AU}. 
\end{proof}

By using the standard  \emph{doubling} approach, Lemma~\ref{lem.AU} implies Theorem~\ref{thm:AU} (see Appendix~\ref{sub:doubling_approach}).

\section{The symmetric favorite machine model}
\label{sec:two_classes_of_machines_case}
This section focuses on the \emph{symmetric} favorite machines problem,
where $\CM{}=\{ M_1,M_2 \}$ ($|M_1| = |M_2| = f$), $\favorite_j \in \{ M_1,M_2 \}$ for each job $j$, and the processing time of job $j$ on a non-favorite machine is $s$ ($\geq 1$) times that on a favorite machine. 
We analyze the competitive ratio of  \GR{} as a function of the parameter $s$. Though \GR has a constant competitive ratio for this problem, another natural algorithm called \GF{} performs better for larger $s$.
At last, a combination of the two algorithms, \GGF{}, obtains a better competitive ratio, and the algorithm is optimal for the two machines case.

\subsection{Greedy Algorithm} 
The next two theorems regard the competitive ratio of the \GR{} algorithm for the symmetric favorite machines case.
\label{sub:greedy4two}
\begin{theorem}
	\label{thm:greedy4two}
	For the \symmetric{$\frac{m}{2}$} case, the \GR{} algorithm has a competitive ratio at most
	\begin{equation*}
	\hat \rho_{\GR{}} :=  \min\left\{1+\left(2- \frac{1}{f}\right) \frac{s^2}{s+1},~s+\left(2- \frac{1}{f}\right)\frac{s}{s+1},~3- \frac{1}{f}\right\} \,,
	\end{equation*}
	where $f= \frac{m}{2}$ and $m$ is the number of machines.
\end{theorem}
Since this upper bound increases in $s$ ($s\geq  1$), we  have the following inequalities:
\[
2 - \frac{1}{m} \leq \hat \rho_{\GR{}} \leq 3 - \frac{2}{m} \, ,
\]
and both bounds can be achieved by the actual competitive ratio $\rho_\GR{}$ of algorithm $\GR{}$. Note that for $s=1$, the lower bound is indeed tight as it corresponds to the analysis of \GR{} on $m$ identical machines \citep{graham1966bounds}. More generally, we have the following result (whose proof is deferred to Appendix~\ref{sub.proof_of_theorem_LB4two} for readability sake): 

\begin{theorem}
	\label{thm:TI}
	The upper bound in Theorem~\ref{thm:greedy4two} is tight ($\rho_\GR{}=\hat\rho_\GR{}$) in each of the following cases:
	\begin{enumerate}
		\item For $f = 1$, $\rho_{\GR{}}= \min\{1+ \frac{s^2}{s+1},~2\}$;
		\item For $f = 2$ and $1 \le s \le 1.605$, $\rho_{\GR{}}=1+ \frac{3s^2}{2(s+1)}$;
		\item For $3 \le f \le \frac{s}{s-1}$ (implying $1 \le s\le 1.5$), $\rho_{\GR{}}=1+(2- \frac{1}{f}) \frac{s^2}{s+1}$;
		\item For $f > \frac{s}{s-1}$ and $1 \le s \le \frac{1+ \sqrt{5}}{2}$,  $\rho_{\GR{}}=s+(2- \frac{1}{f}) \frac{s}{s+1}$;
		\item For $2 \le f < s$, $\rho_{\GR{}}=3- \frac{1}{f}$.
	\end{enumerate}
\end{theorem}

\subsubsection{Proof of Theorem~\ref{thm:greedy4two}}
	Because this is a special case of the general model, we  have
	$\hat{\rho}_\GR{} \le 3- \frac{1}{f}$ from Theorem~\ref{thm.greedyUB} (by recalling that $m=2f$).
	Thus, we just need to prove
	\[
	\frac{C_\GR{}}{C_\opt{}} \le  \min\Big\{1+\big(2- \frac{1}{f}\big) \frac{s^2}{s+1},~s+\big(2- \frac{1}{f}\big)\frac{s}{s+1}\Big\} \,.
	\]
	
	Without loss of generality, we assume that the makespan of the allocation of \GR{} is determined by the  last job $n$.
	Suppose that the first $n - 1$ jobs have been already allocated by \GR{}, and denote 
	$l_\alpha = \min_{i\in M_1} l_i^{(n-1)}, \, L_\alpha = \sum_{i\in M_1} l_i^{(n-1)}, \, l_\beta = \min_{i\in M_2} l_i^{(n-1)},  \, L_\beta = \sum_{i\in M_2} l_i^{(n-1)}$.
	Also let $L_\alpha^\good$ and $L_\beta^\good$ be the total load of good jobs of $L_\alpha$ and $L_\beta$, respectively.
	Observe that $l_\alpha \le \frac{L_\alpha}{f}$ and $l_\beta \le \frac{L_\beta}{f}$.
	Without loss of generality, we assume that machines in $M_1$ are the favorite machines of the last job $n$.
	
	We first give a lower bound of the optimal cost $C_\opt{}$:
	\begin{claim}\label{cla.opt.LB}
		$C_\opt{}  \ge \max \Big\{ \frac{f \cdot s \cdot l_\alpha + f \cdot l_\beta + s \cdot \p{n}}{f\cdot s^2 + f \cdot s} ,\, \frac{ f \cdot l_\alpha + f \cdot s \cdot l_\beta + s^2 \cdot \p{n}}{f\cdot s^2 + f \cdot s} ,\, \p{n} \Big\}$
	\end{claim}
	\begin{proof}
		Denote by $P_\alpha$ and $P_\beta$ the total \emph{minimum processing time} of the jobs which have $M_1$ and $M_2$ as their favorite machines, respectively, i.e.,
		\begin{align}
		\label{eq.P1}
		& P_\alpha = \sum_{j :\ \favorite_j=M_1}\p{j} = L_\alpha^\good{} + \frac{1}{s}(L_\beta - L_\beta^\good{}) + \p{n} \, ,\\ 
		\label{eq.P2}
		& P_\beta = \sum_{j:\ \favorite_j=M_2}\p{j} = L_\beta^\good{} + \frac{1}{s}(L_\alpha - L_\alpha^\good{}) \, .
		\end{align}
		
		A lower bound on the optimal cost can be obtained by considering the following \emph{fractional} assignment. First, allocate all jobs as \emph{good jobs}, i.e., assign $P_\alpha$ to $M_1$ and $P_\beta$ to $M_2$.
		Then, reassign a fraction of them to make all machines to have the same load.  
		We next distinguish two cases:

		\smallskip\noindent
		\textbf{Case 1 ($P_\alpha \ge P_\beta$).}
		In this case, the reassignment is to move $\frac{1}{s+1}(P_\alpha -P_\beta)$ of $P_\alpha$ to machines in $M_2$ so that $$P_\alpha - \frac{1}{s+1} (P_\alpha -P_\beta ) = P_\beta + \frac{s}{s+1} (P_\alpha -P_\beta).$$
		Therefore, along with \eqref{eq.P1} and \eqref{eq.P2}, the optimal cost is at least:
		\begin{equation}\label{eq.opt1.l1a}
		C_\opt{} \ge \frac{s \cdot P_\alpha + P_\beta}{f(s+1)}
		= \frac{L_\alpha +s\cdot L_\beta + s^2 \cdot \p{n} + (s^2-1)L_\alpha^\good{} }{f\cdot s^2 + f \cdot s} \,.
		\end{equation}
		Furthermore, substituting \eqref{eq.P1} and \eqref{eq.P2} into $P_\alpha \ge P_\beta$ we have
		\begin{equation*}
		L_\alpha^\good{} - L_\beta^\good{} \ge \frac{1}{s+1}(L_\alpha - L_\beta) - \frac{s}{s+1}\p{n} \, .
		\end{equation*}
		Therefore,
		\begin{equation}\label{eq.l1a1}
		L_\alpha^\good{} \ge \max \left\{ \frac{1}{s+1}(L_\alpha - L_\beta) - \frac{s}{s+1}\p{n} , \, 0  \right\} \,.
		\end{equation}
		Substituting \eqref{eq.l1a1} into \eqref{eq.opt1.l1a}, we have
		\begin{equation*}
		C_\opt{}  \ge \max \left\{ \frac{s \cdot L_\alpha +  L_\beta + s \cdot \p{n}}{f\cdot s^2 + f \cdot s} ,\, \frac{ L_\alpha + s \cdot L_\beta + s^2 \cdot \p{n}}{f\cdot s^2 + f \cdot s} \right\} \,.
		\end{equation*}
		Along with $C_\opt{}  \ge \p{n}$, $L_\alpha \ge f \cdot l_\alpha$ and $L_\beta \ge f \cdot l_\beta$, we obtain the inequation	of this claim.

		\smallskip\noindent
		\textbf{Case 2 ($P_\alpha < P_\beta$).}
		Similarly to the previous case, we have
		\begin{equation}\label{eq.opt2.l2a}
		C_\opt{} \ge \frac{ P_\alpha + s \cdot P_\beta}{f(s+1)}
		= \frac{s \cdot L_\alpha + L_\beta + s \cdot \p{n} + (s^2-1)L_\beta^\good{} }{f\cdot s^2 + f \cdot s}
		\end{equation}
		In this case, \eqref{eq.P1} and \eqref{eq.P2} imply
		\begin{equation}
		\label{eq.l2-l1}
		L_\beta^\good{} \ge \max \left\{  - \frac{1}{s+1}(L_\alpha-L_\beta) + \frac{s}{s+1}\p{n}, \, 0 \right\} \,,
		\end{equation}
		Substituting \eqref{eq.l2-l1} into\eqref{eq.opt2.l2a}, we have
		\begin{equation*}
		C_\opt{}  \ge \max \left\{  \frac{ L_\alpha + s \cdot L_\beta + s^2 \cdot \p{n}}{f\cdot s^2 + f \cdot s}  ,\, \frac{s \cdot L_\alpha +  L_\beta + s \cdot \p{n}}{f\cdot s^2 + f \cdot s}\right\} \,.
		\end{equation*}
		Along with $C_\opt{}  \ge \p{n}$, $L_\alpha \ge f \cdot l_\alpha$ and $L_\beta \ge f \cdot l_\beta$, we obtain the inequation	of this claim and complete the proof.
	\end{proof}
	
	We next consider the cost of the \GR{} algorithm.
	Recall that job $n$ has $M_1$ as favorite machines.
	After job $n$ is allocated, we have 
	\begin{equation}\label{eq.symgreedy}
	C_\GR{} \le \min \{ l_\alpha + \p{n} , \, l_\beta + s \cdot\p{n} \} \, .
	\end{equation}
	Two cases arise depending on the largest between the two quantities in \eqref{eq.symgreedy}:
	
	\smallskip
	\noindent
	\textbf{Case 1 ($l_\alpha + \p{n} \le l_\beta +  s \cdot \p{n}$).} This case implies
	\begin{gather}
	\label{eq.l2}
	l_\beta \ge  l_\alpha - (s-1) \cdot \p{n} \, ,\\
	\label{eq.C_GR4two}
	C_\GR{} \le  l_\alpha + \p{n} \, .
	\end{gather}
	Therefore, we have 
	\begin{align*}
	\frac{C_\GR{}}{C_{\opt}} & \le \frac{l_\alpha + \p{n}}{\max \Big\{ \frac{f \cdot s \cdot l_\alpha + f \cdot l_\beta + s \cdot \p{n}}{f\cdot s^2 + f \cdot s} ,\, \frac{ f \cdot l_\alpha + f \cdot s \cdot l_\beta + s^2 \cdot \p{n}}{f\cdot s^2 + f \cdot s} ,\, \p{n} \Big\}} \\
	& \le \frac{l_\alpha + \p{n}}{\max \Big\{ \frac{f  (s+1)  l_\alpha  + (f+s-f s)  \p{n}}{f s^2 + f  s} ,\, \frac{ f (s+1) l_\alpha  + (f s + s^2 -f  s^2)  \p{n}}{f s^2 + f  s} ,\, \p{n} \Big\}} \\
	& = \min \left\{ \frac{(f s^2 + f  s)(x + 1)}{f  (s+1)  x  + f+s-f s},~ \frac{(f s^2 + f  s)(x + 1)}{f  (s+1)  x  + f s+s^2 - f s^2} ,~ x+1\right\} \\
	& \le \min \left\{ 1+\big(2- \frac{1}{f}\big)\frac{s^2}{s+1},~ s+\big(2- \frac{1}{f}\big)\frac{s}{s+1}\right\} ,
	\end{align*}
	where the first inequality is by \eqref{eq.C_GR4two} and Claim~\ref{cla.opt.LB};
	the second inequality is by \eqref{eq.l2};
	the third equation is obtained by defining $x:= {l_\alpha}/{\p{n}}$;
	the first term of the last inequality is obtained by the second and third terms of the third equation (one decreases 
	in $x$ and one increases in $x$); similarly, the second term of the last inequality is obtained by the first and third terms of the third equation.
	
	\smallskip
	\noindent
	\textbf{Case 2 ($l_\alpha + \p{n} \ge l_\beta +  s \cdot \p{n}$).} This case implies
	\begin{equation*}
	l_\alpha  \ge l_\beta +  (s-1) \cdot \p{n}  \quad\text{and}\quad
	C_\GR{} \le  l_\beta + s \cdot \p{n} \, .
	\end{equation*}
	Similarly to the previous case, we can obtain
	\[
		\frac{C_\GR{}}{C_{\opt}} \le \min \left\{ 1+\big(2- \frac{1}{f}\big)\frac{s^2}{s+1},~ s+\big(2- \frac{1}{f}\big)\frac{s}{s+1}\right\} \,.
	\]

	The above two cases conclude the proof of the theorem. \qed

\subsection{GreedyFavorite Algorithm}  
We next consider another algorithm called \GF{} which simply assigns each job $j$ to one of its favorite machines in $\favorite_j$.
	\block{Algorithm \GF{}: Assign each job to one of its \emph{favorite} machines, chosen in a \emph{greedy} fashion (minimum load). }
\label{sub:gf}
It turns out that this natural variant of \GR{} performs better  for large $s$.
\begin{theorem}
\label{thm:BAG}
	For \symmetric{$\frac{m}{2}$}, the \GF{} algorithm has a competitive ratio of $2- \frac{1}{f} +\frac{1}{s}$, where $f= \frac{m}{2}$ and $m$ is the number of machines.
\end{theorem}
\begin{proof}
	Note that in \GF{} all jobs are assigned as good jobs.
	Suppose the overall maximum load occurred on machine 1 in $M_1$.
	Moreover, if there is any job executed on machines in $M_2$, we can remove all of it, which will not decrease $\frac{C_{\GF{}}}{C_{\opt}}$.
	Therefore, all jobs have the same favorite machines $M_1$, and \GF{} assigns all of them to $M_1$.
	
	We also use $l_\alpha$ to represent the minimum load over $M_1$ before job $n$ is allocated.
	Denote by $P_\alpha$ the total minimum processing time of the jobs who have $M_1$ as their favorite machines, which is also the total minimum processing time of all the jobs here.
	Obviously,
	\begin{gather}
	P_\alpha \ge f\cdot l_\alpha +\p{n} \, \\
	\label{eq.GF}
	C_{\GF{}} \le l_\alpha+ \p{n} \,.
	\end{gather}

	The optimal schedule can allocate some of the jobs to $M_2$ to balance the load over all machines. 
	Thus, the optimal cost will be at least 
	\begin{align}
	\label{eq.bag.opt1}
	C_{\opt} & \ge \max \left\{ (P_\alpha - x) \cdot \frac{1}{f},~ x \cdot \frac{s}{f}   \right\} 
	\ge \frac{s}{s+1}P_\alpha \cdot \frac{1}{f} 
	\ge \frac{s}{s+1}\big(l_\alpha + \frac{ \p{n}}{f}\big) \, ,
	\end{align}
	where $x$ is the load of jobs that are assigned to $M_2$ to balance the load over all machines.
	
	According to \eqref{eq.GF}, \eqref{eq.bag.opt1} and $C_{\opt} \ge \p{n}$, we have
	\begin{equation}
	\frac{C_{\GF{}}}{C_{\opt}} \le \min \left \{\frac{l_\alpha+\p{n}}{\frac{s}{s+1}(l_\alpha + \frac{ \p{n}}{f})},~ \frac{l_\alpha +\p{n}}{\p{n}} \right \} 
	\le 2- \frac{1}{f}+ \frac{1}{s} \,. 
	\end{equation}
	Thus the upper bound on the competitive ratio  is proved. 
	
	To see that this bound is tight for any $f$ and $s$, consider the following sequence of jobs:
	\begin{align*}
	\textstyle f(f-1) \times (\frac{1}{f},M_1), \quad f \times (\frac{1}{s},M_1), \quad (1,M_1) \, .
	\end{align*}
	According to algorithm \GF{}, all these jobs are assigned to $M_1$ in a greedy fashion, and thus $C_{\GF{}}=\frac{1}{f}(f-1)+ \frac{1}{s}+ 1=2-\frac{1}{f}+ \frac{1}{s}$.
	The optimal solution will instead assign the $f$ jobs $(\frac{1}{s},M_1)$ to $M_2$, thus implying $C_{\opt}=1$.
\end{proof}

\subsection{A better algorithm}
As one can see, the \GR{} is better than \GF{} for smaller $s$, and \GF{} is better for larger $s$.
Thus we can combine the two algorithms to obtain a better algorithm.
\block{Algorithm \GGF{} (GGF): If $s\le s^*$, run \GR{}; otherwise run \GF{}.}
\begin{corollary}For \symmetric{$\frac{m}{2}$}, if $s^* \simeq 1.481$ then
	$\rho_{\textsc{GGF}} \le \min \{2+\frac{s^2+s-2}{s+1} , 2+\frac{1}{s}\} \le 2.675$.
\end{corollary}
\begin{proof}
By Theorem~\ref{thm:greedy4two}, we have 
\begin{equation*}
\rho_{\GR{}} 
 \le s+\left(2- \frac{1}{f}\right)\frac{s}{s+1} 
 \le s+\frac{2s}{s+1} 
 = 2+\frac{s^2+s-2}{s+1} \, .
\end{equation*}
By Theorem~\ref{thm:BAG}, we have 
\[
	\rho_\GF{} \le 2- \frac{1}{f} +\frac{1}{s} \le 2+\frac{1}{s} \,.
\]
Note that if $s\le 1.481$, $2+\frac{s^2+s-2}{s+1} \le 2+\frac{1}{s}$, otherwise $2+\frac{s^2+s-2}{s+1} > 2+\frac{1}{s}$, thus
\[
	\rho_{\textsc{GGF}} \le \min \left\{2+\frac{s^2+s-2}{s+1} , 2+\frac{1}{s} \right\} \le 2.675 \,. \qedhere
\]
\end{proof}

\subsection{Tight bounds for two machines (symmetric $1$-favorite machines)} 
\label{sub:two_machines_case}
In this section, we show that the \textsc{GGF} algorithm is optimal for the symmetric case with \emph{two} machines, i.e., the  \emph{symmetric $1$-favorite machines}.
\begin{theorem}
\label{thm:LB2}
	For symmetric $1$-favorite machines, any deterministic online algorithm has competitive ratio $\rho \ge \min  \left\{ 1+ \frac{s^2}{s+1},~ 1+ \frac{1}{s}  \right\}$.
\end{theorem}
\begin{proof}
	Consider a generic algorithm \alg. 
	Note that we have two machines, $M_1$ contains machine $1$ and $M_2$ contains machine $2$ only.
	Without loss of generality, assume the first job is assigned to machine $1$ and this machine then has load $1$, that is,  job~1 is either
	$	(1,M_1) \text{ or }   (1/s,M_2)	$.
	
	Job 2 is $(s,M_1)$. If \alg{} assigns job 2 to machine $1$, then 
	$C_\alg = 1 + s$ while $C_{\opt} = s$, thus implying $\rho_{\alg} \ge 1 + \frac{1}{s}$. Otherwise, if job~2 is assigned to machine $2$, then a third job $(s+1,M_2)$ arrives. No matter where \alg{} assigns job~3, the cost for \alg{} will be $C_\alg=s^2+s+1$. As the optimum is $C_{\opt}= s+1$, we have $\rho_\alg \ge 1+\frac{s^2}{s+1}$ in the latter case. 
\end{proof}

By combining Theorem~\ref{thm:greedy4two}, Theorem~\ref{thm:BAG}, and Theorem~\ref{thm:LB2}, we obtain the following:
\begin{corollary}
\label{cor:two-machines:tight-bounds}
	For symmetric $1$-favorite machines, if $s^* \simeq 1.481$ then
	$\rho_{\textsc{GGF}}= \min  \{ 1+ \frac{s^2}{s+1},~ 1+ \frac{1}{s} \}\le 1.7549$.
	Therefore,  the
	\textsc{GGF} algorithm is optimal.
\end{corollary}

\section{An extension of our model}
In this section, we discuss a simple extension,
which explains why the instances, where $f$ is small, still have a good competitive ratio.
The main idea is to consider \emph{favorite} machines as the machines which have ``approximately'' the minimal processing time for the job. For example, a job with processing times $(0.99,1,1,1, 2)$ might be considered
to have approximate processing times $(1,1,1,1,2)$. In the latter case, the job has $4$ favorite machines, instead of $1$.

More formally, we consider the following modified algorithm $\hat{A}$ of a generic online algorithm $A$.
For a set of jobs $\CJ{}$, fix a parameter $c\geq 1$ and denote 
$\hat{\favorite}_j := \left\{i: \ \p[i]{j} \leq c \cdot \p{j} \right\}$.
Run algorithm $A$ assuming processing times are
\[
\hatp[i]{j} := \begin{cases}
\p{j} & \text{if } i \in \hat{\favorite}_j \,,\\
\p[i]{j}& \text{otherwise.}
\end{cases}
\]
Note that in the rescaled processing times above the number $\hat{f}$ of favorite machines per job satisfies $\hat{f}\geq f$ and $\p[i]{j} \leq c \cdot \hatp[i]{j}$. Ideally, we would like $\hat f$ as big as possible and $c$ as small as possible, as the following observation indicates. 
\begin{observation}
	If algorithm $A$ is $\rho(f)$-competitive for a certain class of instances, where $f$ denotes the minimum number of favorite machines per job in the input instance, then the modified algorithm $\hat A$ is at most $c \cdot \rho(\hat f)$-competitive on the same class of instances.
\end{observation}

\section{Conclusion and open questions}
This work studies online scheduling for the \emph{favorite} machine model.
Our results are supplements to several classical problems and reveal the relations among them (as indicated in Figure~\ref{fig:problems}).
For the general $f$-favorite machines case, we provide tight bounds on both \GR{} and \AU{} algorithms and show that the latter is the best-possible online algorithm.
To some extent, the key factor $f$ in our model captures some of the main features that make the model perform well or badly: low or high competitive ratio.
In particular, when $f=1$, the model is exactly the \emph{unrelated} machines;
when $f=m$, the model is exactly the \emph{identical} machines.
Finally, the analysis of symmetric favorite machines allows a direct comparison with the two related machines.

\section*{Acknowledgments} 
This work was supported by the National Natural Science Foundation of China [71601152]; and the China Postdoctoral Science Foundation [2016M592811].
Part of this work has been done while the first author was visiting ETH Zurich.

\bibliographystyle{plainnat}
\bibliography{mybib}

\appendix
\section{Postponed proofs} 
\label{app.proof}

\subsection{The potential function is non-increasing (for  proof of Lemma~\ref{lem.AU})} 
\label{sub:non-increasing_potential_function}
Recall that the potential function is defined as $\Phi(j) = \sum_{i=1}^m a^{\tload[i]{j}} (\gamma - \toload[i]{j})$.
Assume that job $j$ is assigned to machine $i'$ by algorithm \AU{} and to machine $i$ by the optimal schedule, i.e., $\tload[i']{j} = \tload[i']{j-1} + \tp[i']{j}$ and $\toload[i]{j} = \toload[i]{j-1} + \tp[i]{j}$.  Then we have
	\begin{align*}
	\Phi(j) - \Phi(j-1) & = (\gamma - \toload[i']{j-1})(a^{\tload[i']{j}} - a^{\tload[i']{j-1}}) -  a^{\tload[i]{j-1}} \tp[i]{j} \\
	& \le \gamma(a^{\tload[i']{j-1} + \tp[i']{j}} - a^{\tload[i']{j-1}}) -  a^{\tload[i]{j-1}} \tp[i]{j} \\
	& \le \gamma(a^{\tload[i]{j-1} + \tp[i]{j}} - a^{\tload[i]{j-1}}) -  a^{\tload[i]{j-1}} \tp[i]{j} & (\text{by~} \Delta_{i'} \le \Delta_i)\\
	& = a^{\tload[i]{j-1}} ( \gamma (a^{\tp[i]{j}}-1)-\tp[i]{j})\ . 
	\end{align*}
	By taking $a = 1+ 1/ \gamma$, we get $\gamma (a^{\tp[i]{j}}-1)-\tp[i]{j} \le 0$ since $0 \le \tp[i]{j} \le 1$,
	so that the potential function is non-increasing.

\subsection{Doubling approach (proof of Theorem~\ref{thm:AU})} 
\label{sub:doubling_approach}
Let $\rho$ be the competitive ratio of \AU{} when the optimal cost $C_\opt{}$ is known.
By using \emph{doubling} approach one can easily get a $4 \rho$-competitive algorithm for the case optimal cost is not known. This approach has been used in \cite{aspnes1997line}. We report the details below for completeness. 

We run \AU{} in phases,
and let $\Lambda_i$ be the estimation of $C_\opt{}$ at the beginning of phase $i$.
Initially (beginning of phase~1) when the first job arrives, let $\Lambda_1$ be the minimum processing time of the first job.
Whenever the makespan exceeds $\rho$ times the current estimation, $\rho \Lambda_{i}$, the current phase $i$ ends and the next phase $i+1$ begins with doubled estimation $\Lambda_{i+1} = 2\Lambda_{i}$ as the new estimation of the $C_\opt{}$ to run \AU{}.
During a single phase, jobs are assigned independently of the jobs assigned in the previous phases.
It is easy to see that this approach increases the competitive ratio $\rho$ by at most a multiplicative factor $4$ (a factor of 2 due to the load in all but the last phase, and another factor of 2 due to imprecise estimation of $C_\opt{}$).

More in detail, each phase $i$ can increase the load of every machine by at most  $\rho \Lambda_i$. If $u$ denotes the number of phases, then the final makespan will be no more than $\rho \sum_{i=1}^{u} \Lambda_i$.
Note that $\sum_{i=1}^{u} \Lambda_i = (1 + \frac{1}{2} + \dots + \frac{1}{2^{u-1}})\Lambda_u = (2- \frac{1}{2^{u-1}})\Lambda_u$, since $\Lambda_{i+1} = 2\Lambda_{i}$.
We also have $\Lambda_u = 2\Lambda_{u-1} < 2 C_\opt{}$, because $\Lambda_{u-1} < C_\opt{}$ (otherwise in phase $u-1$ the makespan will not exceeds $\rho \Lambda_{u-1}$ according to Lemma~\ref{lem.AU}).
Thus we have $\rho \sum_{i=1}^{u} \Lambda_i = (2- \frac{1}{2^{u-1}})\rho \Lambda_u < 4 \rho C_\opt{}$.

\subsection{Proof of Theorem \ref{thm:TI}}\label{sub.proof_of_theorem_LB4two}

In some of the proofs we shall make use of the following initial set of ``small'' jobs: 
\begin{equation}
	\label{eq:initial-tiny-jobs}
	\underbrace{f \times (\epsilon,M_1),\, f \times (\frac{2\epsilon}{s},M_1),\, f \times (\frac{2\epsilon}{s},M_2),\, f \times (\frac{2\epsilon}{s},M_1),\, f \times (\frac{2\epsilon}{s},M_2) \ldots}_{t/\epsilon \text{ blocks} }
\end{equation}
where the total number of jobs is $f \cdot t/\epsilon$, and $\epsilon$ is chosen so that $t/\epsilon$ is integer.

According to the algorithm \GR{},  only the first  first $f$ jobs of length $\epsilon$ are assigned as good jobs, while all other jobs are assigned as bad jobs. Moreover, all machines in each class will have the same load $t$ and $t- \epsilon$.
These jobs can be redistributed  to the machines in order to built arbitrary load (up to some arbitrarily small additive $\epsilon$). Taking $\epsilon \to 0$, we can obtain the following result. 
\begin{lemma}
\label{lem.bad}

At the beginning of a schedule by algorithm \GR{}, 
if $s<2$, 
each machine can have a load of $t$ so that
all jobs executed during $[0,t]$ are bad jobs, 
and each bad job is extremely ``small'' so that they can be redistributed to create an arbitrary load on any machine.

\end{lemma}

\begin{proof}[Proof of Theorem \ref{thm:TI}]
We give five instances each of them resulting in a lower bound for the corresponding case.

\smallskip\noindent
\textbf{Case 1 ($f=1$).}\label{Example 1}
If $1\le s\le \frac{1+\sqrt{5}}{2}$, the jobs sequence is $(\frac{1}{s+1},M_2)$, $(\frac{s}{s+1},M_2)$ and $(1,M_1)$.
The \GR{} algorithm assigns the first job to machine 2, and the last two jobs to machine 1, which leads to $C_{\GR{}}=1+ \frac{s^2}{s+1}$. In optimal schedule all jobs are allocated as good jobs, i.e., $C_{\opt}=1$.

If $s > \frac{1+\sqrt{5}}{2}$, the jobs sequence is $(\frac{s-1}{s},M_2)$, $(\frac{1}{s},M_2)$ and $(1,M_1)$.
The \GR{} assigns the first job to machine 2, and the last two jobs to machine in 1, which leads to $C_{\GR{}}=2$. Again, in optimal schedule all jobs are allocated as good jobs, i.e., $C_{\opt}=1$.

Therefore, $\rho_{\GR{}}=\min\{1+ \frac{s^2}{s+1},~2\}$ is tight for any $s \ge 1$.

\smallskip\noindent
\textbf{Case 2 ($f = 2$ and $1 \le s \le 1.605$).}\label{Example 2}
Let  $l_\alpha=\frac{3s^2}{2(s+1)}$, $S_6=\sum_{i=1}^6 (s-1)^i=\frac{s-1}{2-s}(1-(s-1)^6)$ and $l_\beta^\good{}=\frac{2-s}{2(s+1)}$. The sequence of jobs corresponds to the following three steps:

\smallskip\noindent
\emph{Step~1:}
We use Lemma~\ref{lem.bad} to let each machine have and initial load  $l_\alpha - S_6 - l_\beta^\good{} + (s-1)^6$ of bad jobs, where $l_\alpha - S_6 - l_\beta^\good{} + (s-1)^6 \ge 0$ due to $1 \le s \le 1.605$.

\smallskip\noindent
\emph{Step~2:}
Four jobs arrive: $2\times (l_\beta^\good{} - (s-1)^6,M_2)$ and $2\times (\frac{l_\beta^\good{} - (s-1)^6}{s},M_2)$.
According to \GR{}, the first two jobs will be assigned to machine 3 and 4 respectively as good jobs.
But the last two jobs will be assigned to machine 1 and 2 as bad jobs.
At this point, all the four machines have the same load  $l_\alpha - S_6$.

\smallskip\noindent
\emph{Step~3:}
This sequence of jobs arrive: $2 \times ((s-1)^6,M_2)$, $2 \times ((s-1)^5,M_2)$, $2 \times ((s-1)^4,M_1)$, $2 \times ((s-1)^3,M_2)$, $2 \times ((s-1)^2,M_1)$, $2 \times (s-1,M_2)$.
According to \GR{}, the first two jobs are allocated as good jobs, while the others are allocated as bad jobs.
At this point, the load of machine 1 and 2 is $l_\alpha$, while machine 3 and 4 have load $l_\alpha-(s-1)$.

\smallskip\noindent
\emph{Step~4:}
Job $(1,M_1)$ arrives, which will be assigned to machine 3 as a bad job.
Therefore, $C_{\GR{}}=l_\alpha+1=1+\frac{3s^2}{2(s+1)}$. 
The optimal schedule is to assign all jobs as good jobs.
By calculation we  have $C_{\opt}=1$.
Thus, $\rho_{\GR{}}=1+\frac{3s^2}{2(s+1)}$ is tight for $1 \le s \le 1.605$.

\smallskip\noindent
\textbf{Case 3 ($3 \le f \le \frac{s}{s-1}$).}\label{Example 3}
Let $l_\alpha=(2-1/f) \frac{s^2}{s+1}$, $a_i=(s-1)^i$, $S_u=\sum_{i=1}^u a_i= \frac{s-1}{2-s}(1-a_u)$ and $l_\beta^\good{}=\frac{f+s-f\cdot s}{f(s+1)}$, where $u$ is even number.
Suppose $3 \le f\le \frac{s}{s-1+(s+1)a_u}$ and $1\le s \le 1.5$.
Note that when $u \to \infty$, we have $a_u \to 0$, i.e. $3 \le f\le \frac{s}{s-1}$.

\smallskip\noindent
\emph{Step~1:}
We use Lemma~\ref{lem.bad} to let each machine have an initial load $l_\alpha - S_u - l_\beta^\good{} + a_u$ of bad jobs, where $l_\alpha - S_u - l_\beta^\good{} + a_u \ge 0$ due to $3 \le f\le \frac{s}{s-1+(s+1)a_u}$ and $1\le s \le 1.5$.

\smallskip\noindent
\emph{Step~2:}
These $2f$ jobs arrive: $f\times (l_\beta^\good{} - a_u,M_2)$ and $f\times (\frac{l_\beta^\good{} - a_u}{s},M_2)$.
According to \GR{}, the first $f$ jobs will be assigned to $M_2$ as good jobs with one machine each,
while the last $f$ jobs will be assigned to $M_1$ as bad jobs with one machine each.
At this point, all the $2f$ machines has the same load of $l_\alpha - S_u$.

\smallskip\noindent
\emph{Step~3:}
This sequence of jobs arrives: $f \times (a_u,M_2)$, $f \times (a_{u-1},M_2)$, $f \times (a_{u-2},M_1)$, $f \times (a_{u-3},M_2)$, $f \times (a_{u-4},M_1)$,\ldots, $f \times (a_2,M_1)$, $f \times (a_1,M_2)$.
According to \GR{}, the first $f$ jobs will be allocated as good jobs while the others as bad jobs.
At this point, each machine in $M_1$ has load  $l_\alpha$, and each machine in $M_2$ has load  $l_\alpha-(s-1)$.

\smallskip\noindent
\emph{Step~4:}
Job $(1,M_1)$ arrives, which is assigned to one machine in $M_2$ as a bad job.
Therefore, $C_{\GR{}}=l_\alpha+1=1+(2-1/f)\frac{s^2}{s+1}$. 

\smallskip
The optimal schedule will allocate all jobs as good jobs. 
We next give such an optimal schedule to show that $C_{\opt}=1$ is achievable:
	
	\noindent(\emph{Step~1 jobs.}) All the jobs in Step 1 will be allocated as good jobs in optimal schedule, meaning $f \times\frac{l_\alpha-S_u-l_\beta^\good{}+a_u}{s}$ for each of $M_1$ and $M_2$;
	
	\noindent(\emph{Step~2 jobs.}) All the jobs in Step 2 will be assigned to $M_2$;
	
	\noindent(\emph{Step~3 jobs.}) All the jobs in Step 3 will be allocate as good jobs, meaning $f \times (a_2+a_4+ \cdots + a_{u-2})$ will be assigned to $M_1$, while the rest of them to $M_2$;
	
	\noindent(\emph{Step~4 jobs.}) The last job in Step 4 will be assigned to $M_1$.

For the jobs allocated to $M_1$, notice that every 2 jobs of $f \times (a_2+a_4+ \cdots + a_{u-2})$ should be assigned to one machine, i.e., the load of some $\lfloor \frac{f}{2} \rfloor$ machines are all $2 \times (a_2+a_4+ \cdots + a_{u-2})$,
where $2 \times (a_2+a_4+ \cdots + a_{u-2})=\frac{2 a_u -2 (s-1)^2}{s(s-2 )}<1$ since $3 \le f\le \frac{s}{s-1+(s+1)a_u}$ and $1\le s \le 1.5$.
Then the other jobs can be easily arranged within time 1, since the jobs in Step 1 are all ``small'' jobs.

For the jobs allocated to $M_2$,
the jobs can be equally divided into $f$ parts with each part has $\frac{l_\alpha-S_u-l_\beta^\good{}+a_u}{s}+ (l_\beta^\good{} - a_u) + \frac{l_\beta^\good{} - a_u}{s} + (a_1+a_3+ \cdots + a_{u-1}+a_u){=1}$.
Thus all machines in $M_2$ also have the same load of 1.
Therefore, $C_{\opt}=1$.

Thus, $\rho_{\GR{}}=1+(2-1/f)\frac{s^2}{s+1}$ is tight for $3 \le f\le \frac{s}{s-1+(s+1)a_u}$ and $1\le s \le 1.5$.
Taking $u \to \infty$, we have $a_u \to 0$, i.e. $3 \le f\le \frac{s}{s-1}$.

\smallskip\noindent
\textbf{Case 4 ($f > \frac{s}{s-1}$ and $1 \le s \le \frac{1+ \sqrt{5}}{2}$).}\label{Example 4}
Let $l_\alpha=s + \frac{f\cdot s-f-s}{f(s+1)}$, $a_i=(s-1)^i$, $S_u=\sum_{i=1}^u a_i= \frac{s-1}{2-s}(1-a_u)$ and $l_\alpha^\good{}=\frac{f\cdot s -f -s}{f(s+1)}$, where $u$ is odd number.
Suppose $ f > \frac{s}{s-1-(s+1)a_u}$ and $1\le s \le \frac{1+ \sqrt{5}}{2}$.
When $u \to \infty$, we have $a_u \to 0$, i.e. $ f > \frac{s}{s-1-(s+1)a_u}$.

\smallskip\noindent
\emph{Step~1:}
We use Lemma~\ref{lem.bad} to let each machine have an initial load $l_\alpha - S_u $ of bad jobs, where $l_\alpha - S_u > 0$ due to $ f > \frac{s}{s-1-(s+1)a_u}$ and $1\le s \le \frac{1+ \sqrt{5}}{2}$.

\smallskip\noindent
\emph{Step~2:}
This sequence of jobs arrives: $f \times (a_u,M_1)$, $f \times (a_{u-1},M_1)$, $f \times (a_{u-2},M_2)$, $f \times (a_{u-3},M_1)$, $f \times (a_{u-4},M_2)$,\ldots, $f \times (a_2,M_1)$, $f \times (a_1,M_2)$.
According to \GR{}, the first $f$ jobs will be allocated as good jobs while the others as bad jobs.
Up to now each machine in $M_1$ has load of $l_\alpha$, and each machine in $M_2$ has load of $l_\alpha-(s-1)$.

\smallskip\noindent
\emph{Step~3:}
Job $(1,M_1)$ arrives, which will be assigned to one machine in $M_2$ as a bad job.
Therefore, $C_{\GR{}}=l_\alpha+1=s + (2-1/f)\frac{s}{s+1}$.

For the optimal cost, we will show a schedule so that $C_{\opt}=1$.
Part of the jobs in Step~1 will be allocated as bad jobs in optimal schedule, specifically
some jobs with total minimum processing time $f \times \frac{l_\alpha^\good{}-a_u}{s}$ having $M_2$ as their favorite machine set will be allocated to $M_1$ as bad jobs,
i.e., $M_1$ will have jobs with total load $f \times \frac{l_\alpha-S_u}{s} + f \times (l_\alpha^\good{}-a_u)$ while $M_2$ will have jobs with total load $f \times (\frac{l_\alpha-S_u}{s}-\frac{l_\alpha^\good{}-a_u}{s})$;
all the jobs in Step~2 and 3 will be allocate as good jobs, meaning $f \times (a_u+a_{u-1}+a_{u-3}+\cdots+a_2 )+1$ will be assigned to $M_1$, while the rest of them to $M_2$.
To sum up, $M_1$ have jobs with total load $f \times \frac{l_\alpha-S_u}{s} + f \times (l_\alpha^\good{}-a_u) + f \times (a_u+a_{u-1}+a_{u-3}+\cdots+a_2 )+1 = f$, while $M_2$ have $f \times (\frac{l_\alpha-S_u}{s}-\frac{l_\alpha^\good{}-a_u}{s})+ f \times (a_{u-2}+a_{u-4}+ \cdots + a_1)=f$.

Then we give a schedule so that each machine has the same load 1. For the jobs allocated to $M_1$, we first arrange the $f \times (a_u+a_{u-1}+a_{u-3}+\cdots+a_2)$ and 1.
The job with length 1 will be assigned to machine 1, and jobs $(f-1) \times (a_u+a_{u-1}+a_{u-3}+\cdots+a_2)$ will be assigned to the remaining $f-1$ machines with $(a_u+a_{u-1}+a_{u-3}+\cdots+a_2)$ each.
The remaining $(a_u+a_{u-1}+a_{u-3}+\cdots+a_2)$ will be divided into 2 parts, $a_2$ assigned to machine 2 and $(a_u+a_{u-1}+a_{u-3}+\cdots+a_4)$ to machine 3.
Note that $(a_u+a_{u-1}+a_{u-3}+\cdots+a_2)+a_2<1$ and $2(a_u+a_{u-1}+a_{u-3}+\cdots+a_2)-a_2<1$, due to $ f > \frac{s}{s-1-(s+1)a_u}$, $1\le s \le \frac{1+ \sqrt{5}}{2}$ and $a_u \le s-1$.
Till now, no machine has load more than 1,
and the remaining jobs are all ``small'' jobs which can be arbitrary divided and assigned to make every machine with load 1.
For the jobs allocate to $M_2$, they can be equally divided into $f$ parts with each size 1.
Therefore, all machines have the same load 1.

Thus, $\rho_{\GR{}}=s+(2-1/f)\frac{s}{s+1}$ is tight for $ f > \frac{s}{s-1-(s+1)a_u}$ and $1\le s \le \frac{1+ \sqrt{5}}{2}$.
Taking $u \to \infty$, we have $a_u \to 0$, i.e. $ f > \frac{s}{s-1}$.

\smallskip\noindent
\textbf{Case 5 ($2 \le f < s$).}\label{Example 5}
Consider this jobs sequence: $f \times (1-1/s,M_2)$, $f \times (1/s,M_2)$, $f(f-1) \times (1/f,M_1)$ and $(1,M_1)$.

According to algorithm \GR{}, the first $f$ jobs will be assigned to $f$ machines in $M_2$ respectively, so that each machine in $M_2$ has load $1-1/s$.
Then the next $f$ jobs will be assigned to $f$ machines in $M_1$ respectively, so that each machine in $M_1$ has load 1.
In terms of the $f(f-1)$ jobs with length $1/f$, all of them will be assigned to machines in $M_1$ with $f-1$ job each machine. Note that none of the $f(f-1)$ jobs will go to $M_2$, since $1- \frac{1}{s}+ \frac{s}{f} > 1+ \frac{f-1}{f}$. Now all machines in $M_1$ have the same load $2- \frac{1}{f}$.
At last, the final job with length 1 will be assigned to one machine in $M_1$, since $2-\frac{1}{f}+1 < 1- \frac{1}{s}+s$. Therefore, $C_{\GR{}}=3- \frac{1}{f}$.

For the optimal cost, it is easy to have $C_{\opt}=1$ by assigning each job to its favorite class of machines.

Therefore, $\rho_{\GR{}}=3-1/f$ is tight for $2 \le f < s$.
\end{proof}

\end{document}